\documentclass[english,a4paper]{nolineno-lipics-v2021}

\pdfoutput=1 
\hideLIPIcs  

\usepackage[utf8]{inputenc}
\usepackage{amsmath}
\usepackage{amssymb}
\usepackage{hyperref}
\usepackage{xcolor}
\usepackage{listings}
\usepackage{complexity}

\usepackage{float}
\usepackage{graphicx}

\input{pdffig.sty}

\newtheorem{Claim}{Claim}

\def\NN{\mathbb{N}}
\def\cC{\cal{C}}

\def\bv{\bar{v}}
\def\bG{\bar{G}}

\title{Linear Size Universal Point Sets for Classes of Planar Graphs}

\author{Stefan {Felsner}}{Institute of Mathematics, Technische Universität Berlin, Germany \and\url{https://page.math.tu-berlin.de/~felsner/}}{felsner@math.tu-berlin.de}{https://orcid.org/0000-0002-6150-1998}{DFG Project FE~340/12-1.}

\author{Hendrik Schrezenmaier}{Institute of Mathematics, Technische Universität Berlin, Germany \and\url{https://page.math.tu-berlin.de/~schrezen/}}{schrezen@math.tu-berlin.de}{https://orcid.org/0000-0002-1671-9314}{DFG Project FE~340/12-1.}

\author{Felix Schr\"{o}der}{Institute of Mathematics, Technische Universität Berlin, Germany \and\url{https://page.math.tu-berlin.de/~fschroed/}}{fschroed@math.tu-berlin.de}{https://orcid.org/0000-0001-8563-3517}{}

\author{Raphael Steiner}{Institute of Theoretical Computer Science, Department of Computer Science, ETH Z\"{u}rich, Switzerland\and\url{https://sites.google.com/view/raphael-mario-steiner/about-me/}}{raphaelmario.steiner@inf.ethz.ch}{https://orcid.org/0000-0002-4234-6136}{DFG-GRK 2434 Facets of Complexity and an ETH Z\"{u}rich Postdoctoral Fellowship}

\authorrunning{S. Felsner, H. Schrezenmaier, F. Schr\"oder and R. Steiner}

\Copyright{S. Felsner, H. Schrezenmaier, F. Schr\"oder and R. Steiner} 

\date{\today}

 \EventEditors{Erin W. Chambers and Joachim Gudmundsson}
\EventNoEds{2}
\EventLongTitle{39th International Symposium on Computational Geometry
(SoCG 2023)}
\EventShortTitle{SoCG 2023}
\EventAcronym{SoCG}
\EventYear{2023}
\EventDate{June 12--15, 2023}
\EventLocation{Dallas, Texas, USA}
\EventLogo{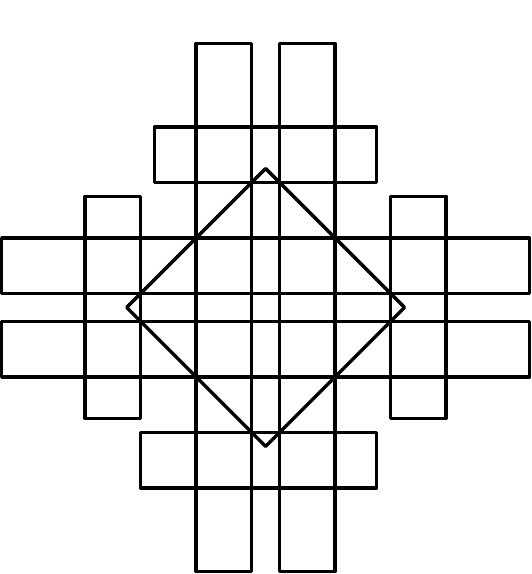}
\SeriesVolume{258}
\ArticleNo{XX} 

\ccsdesc[100]{Graphs and Surfaces}

\acknowledgements{
We are highly indebted to Henry Förster, Linda Kleist, 
Joachim Orthaber and Marco Ricci due to discussions during GG-Week 2022
resulting in a solution to the problem of separating 2-cycles in our proof for subcubic graphs.}
\begin{document}
\color{black}
\maketitle

\begin{abstract}
A finite set $P $ of points in the plane is 
$n$-universal with respect to a class $\mathcal{C}$ of 
planar graphs if every $n$-vertex graph in
$\mathcal{C}$ admits a crossing-free straight-line
drawing with vertices at points of $P$.

For the class of all planar graphs the best known upper
bound on the size of a universal point set is quadratic
and the best known lower bound is linear in~$n$.

Some classes of planar graphs are known to admit universal
point sets of near linear size, however, there are no
truly linear bounds for interesting classes beyond
outerplanar graphs.

In this paper, we show that there is a universal point
set of size $2n-2$ for the class of bipartite planar graphs
with $n$ vertices. The same point set is also universal
for the class of $n$-vertex planar graphs
of maximum degree~$3$. The point set used for the
results is what we call an exploding double chain, and we prove that this point set
allows planar straight-line embeddings of many more planar graphs, namely of all subgraphs of  planar graphs
admitting a one-sided Hamiltonian cycle.

The result for bipartite graphs also implies that
every $n$-vertex plane graph has a 1-bend drawing all
whose bends and vertices are contained in a specific
point set of size $4n-6$, this improves a bound of
$6n-10$ for the same problem by Löffler and Tóth.
  
	\keywords{Graph drawing, \and Universal point set, \and One-sided Hamiltonian, 
    \and 2-page book embedding, \and Separating decomposition, \and Quadrangulation, \and 2-tree, \and Subcubic planar graph}
\end{abstract}

\section{Introduction}
Given a family $\mathcal{C}$ of planar graphs and a positive integer $n$, a
point set $P \subseteq \mathbb{R}^2$ is called an \emph{$n$-universal point
	set} for the class $\mathcal{C}$ or simply \emph{$n$-universal} for
$\mathcal{C}$ if for every graph $G \in \mathcal{C}$ on~$n$ vertices there
exists a straight-line crossing-free drawing of $G$ such that
every vertex of $G$ is placed at a point of $P$.

To determine the minimum size of universal sets for classes of planar graphs
is a fundamental problem in geometric graph theory, see
e.g.~Problem~\cite{OPG} in the Open Problem Garden. More specifically, the
quest is for good bounds on the minimum size $f_{\mathcal{C}}(n)$ of an
$n$-universal point set for a class $\mathcal{C}$. 

Schnyder~\cite{Schnyder1990} showed that for $n \ge 3$ the
$[n-1]\times [n-1]$-grid forms an $n$-universal point set for planar graphs,
even if the combinatorial embedding of the planar graph is prescribed. This
shows that $f(n):=f_{\mathcal{P}}(n)\leq n^2\in O(n^2)$, where~$\mathcal{P}$ is the
class of all planar graphs.
 Asymptotically, the quadratic upper bound on $f(n)$
remains the state of the art. Only the multiplicative constant in this
bound has seen some improvement, the current upper bound is
$f(n) \le \frac{1}{4}n^2+O(n)$ by Bannister et
al.~\cite{BannisterCDE2014}.
For several subclasses $\mathcal{C}$ of planar
graphs, better upper bounds are known: A classical result by Gritzmann et
al.~\cite{GritzmannMoharPachPollack1991} is that every outerplanar $n$-vertex
graph embeds straight-line on \emph{any} set of $n$ points in general
position, and hence $f_{\textrm{out-pl}}(n)=n$. Near-linear upper bounds of
$f_\mathcal{C}(n)=O(n\;\text{polylog}(n))$ are known for $2$-outerplanar
graphs, simply nested graphs, and for the classes of bounded
pathwidth~\cite{Angelini2018,BannisterCDE2014}. Finally, for the class
$\mathcal{C}$ of planar $3$-trees (also known as Apollonian networks or
stacked triangulations), $f_\mathcal{C}(n)=O(n^{3/2}\log n)$ has been proved
by Fulek and T\'{o}th~\cite{FulekToth2013}.

As for lower bounds, the trivial bounds $n \le f_\mathcal{C}(n) \le f(n)$ hold
for all $n \in \mathbb{N}$ and all planar graph classes $\mathcal{C}$. The
current lower bound $f(n) \ge 1.293n-o(n)$ from~\cite{JGAA-529} has
been shown using planar $3$-trees, we refer
to~\cite{CardinalHoffmannKusters2015,Chrobak1989,DeFraysseixPachPollack1990,Kurowski2004}
for earlier work on lower bounds.

Choi, Chrobak and Costello~\cite{Choi} recently proved
that point sets chosen uniformly at random from the unit square must have size
$\Omega(n^2)$ to be universal for $n$-vertex planar graphs with high
probability. This suggests that universal point sets of size $o(n^2)$ -if they
exist- will not look nice, e.g., they will have a large ratio between
shortest and largest distances.


In this paper we study a specific ordered point set $H$ (the exploding double
chain)
and denote the initial piece of size $2n-2$ in $H$ as $H_n$.
Let $\cC$ be the class of all planar graphs $G$ which have a plane
straight-line drawing on the point set~$H_n$ where $n=|V(G)|$.  That is, $H_n$
forms an $n$-universal point set for $\cC$.

A graph is {POSH} (partial one-sided Hamiltonian) if it is a spanning subgraph
of a graph admitting a plane embedding with a one-sided Hamiltonian
cycle (for definitions see Section~\ref{sec:good}).
Triangulations with a one-sided Hamiltonian cycle have been studied before by
Alam et al.~\cite{AlamBFKKU13} in the context of cartograms. They conjectured
that every plane 4-connected triangulation has a one-sided Hamiltonian
cycle. Later Alam and Kobourov~\cite{AlamK12} found a plane 4-connected
triangulation on 113 vertices which has no one-sided Hamiltonian cycle.

Our main result (Theorem~\ref{thm:embedding}) is that
every POSH graph is in $\cC$.  We let 
\[\cC' := \{ G : G\text{ is POSH}\}.\]

Theorem~\ref{thm:embedding} motivates further study of $\cC'$. On the
positive side we show that every bipartite plane graph is POSH (proof in
Section~\ref{sec:bipartite}). We proceed to use the construction for bipartite
graphs to show that subcubic planar graphs have a POSH embedding in
Section~\ref{sec:cubic}. On the negative side, we also show
that not all $2$-trees are POSH. We conclude with some conjectures and open
problems in Section~\ref{sec:conclusion}.

An exploding double chain was previously used
by L{\"{o}}ffler and T{\'{o}}th~\cite{LofflerT15}.  They show that every
planar graph with $n$ vertices has a 1-bend drawing on a subset $S_n$ of $H$
with $|S_n|=6n-10$. Our result about bipartite graphs implies a
better bound: 

\begin{corollary}
    There is a point set $P=H_{2n-2}$ of size $4n-6$ such that every $n$-vertex planar graph 
    admits a 1-bend drawing with bends and vertices on $P$.
\end{corollary}
\begin{proof}
  The dual of a plane triangulation is a bridgeless 3-regular graph of $2n-4$ vertices; 
  it has a perfect matching by Petersen's Theorem \cite{Petersen1891}. 
  Hence, subdividing at most $n-2$ edges can make any planar graph on $n$ vertices bipartite. 
  Thus $H_{n+n-2}$ of size $2(n+n-2)-2 = 4n-6$ is sufficient
  to accomodate 1-bend drawings of all $n$-vertex planar graphs.
\end{proof}  
Universality for 1-bend and 2-bend drawings 
with no restriction on the placement of bends
has been studied by Kaufmann and Wiese~\cite{KaufmannW02}, they show that
every~$n$-element point set is universal for 2-bend drawings of planar graphs.

\section{The point set and the class of POSH graphs}\label{sec:good}

In this section we define the exploding double chain $H$ and the class
$\cC'$ of POSH graphs and show that for every $n \ge 2$ the initial part $H_n$
of size $2n-2$ of $H$ is $n$-universal for $\cC'$.

A sequence $(y_i)_{i\in \NN{}}$ of real numbers satisfying $y_1=0$, $y_2=0$ is
\emph{exploding} and 
the corresponding point set $H = \{p_i,q_i|{i \in \NN}\}$, where $p_i=(i,y_i), q_i=(i,-y_i)$, 
is an \emph{exploding double chain}, if for all $n\in\NN$, $y_{n+1}$ is large enough that 
all intersections of lines going through two points of $H_n=\{p_i,q_i|{i \in [n]}\}$ 
with the line $x=n+1$ lie strictly between $y_{n+1}$ and $-y_{n+1}$. 
It is $p_1=q_1$ and $p_2=q_2$, thus $|H_n|=2n-2$. Figure~\ref{fig:exploding} shows $H_6$.
  This fully describes the order type of the exploding double chain. 
  Note that the coordinates given here can be made integers, 
  but the largest coordinate of $H_n$ is exponential in $n$, which is unavoidable for the order type. 
  However, the ratio of largest to smallest distance does not have to be: 
  We can alter the construction setting $y_i=i$, but letting the $x$-coordinates grow slowly enough
  as to achieve the same order type, but with a linear ratio.
  
An explicit construction of a point set $H$ in this order type is given now.

A sequence $Y=(y_i)_{i\ge 1}$ of real numbers satisfying $y_1=0$, $y_2=0$, and
$y_{i+1}> 2y_{i}+y_{i-1}$ for all $i \ge 2$ is exploding. Note that if
$\alpha > 1+\sqrt{2}$, then $y_1=y_2=0$ and $y_i=\alpha^{i-3}$ for $i\geq 3$ is an
exploding sequence, e.g. $\alpha =3$. Given an exploding sequence $Y$ let $P(Y)=(p_i)_{i\ge 1}$
be the set of points with $p_i=(i,y_i)$ and let $\bar{P}(Y) = (q_i)_{i\ge 1}$
be the set of points with $q_i=(i,-y_i)$, i.e., the point set 
reflected at the $x$-axis.

Let $H=H(Y)$ for some exploding sequence $Y$.
For two points $p$ and $q$ let $H(p,q)$ be the set of points of $H$ in the open
right half-plane of the directed line $\overrightarrow{pq}$. 
Note that\footnote{In cases where $i$ or $j$ are in $\{1,2\}$ the following may list
	one of the two points defining the halfspace with its second name as member of the
	halfspace. For correctness such listings have to be ignored.}
$$
H(p_i,q_j) =
\begin{cases}
	(p_k)_{k \leq j} \cup (p_k)_{k > i} \cup (q_\ell)_{\ell < j} & 
	\textrm{if}\quad i > j \\
	(p_k)_{k < i} \cup (q_\ell)_{\ell < i}   & \textrm{if}\quad i = j \\
	(p_k)_{k < i} \cup (q_\ell)_{\ell \leq i} \cup (q_\ell)_{\ell > j} & 
	\textrm{if}\quad i < j
\end{cases}
$$
Moreover, if $i< j$ then $H(q_i,q_j) = H(p_i,q_j) \setminus \{q_i\}$ and if 
$i > j$ then
$H(p_i,p_j) = H(p_i,q_j) \setminus \{p_j\}$.
These sidedness conditions characterize the order type of the exploding double chain.
   \calc_figscale{20}
    \begin{figure}[htb]
    \centerline{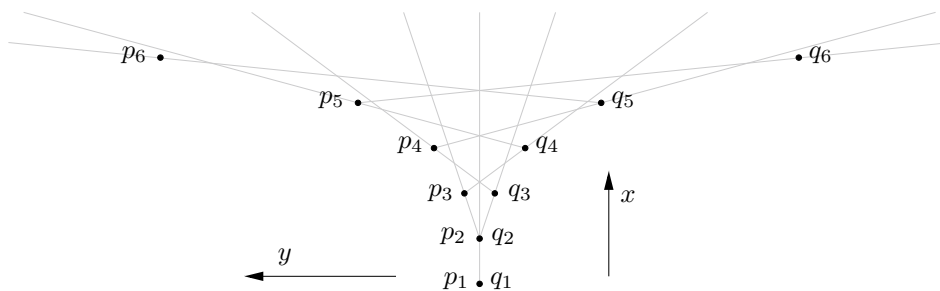}
    \caption{ An example of a point set $H_6$ in a rotated 
	coordinate system.\label{fig:exploding}}
    \end{figure}
    

      A plane graph $G$ has a \emph{one-sided Hamiltonian cycle with special
      edge $vu$} if it has a Hamiltonian cycle $(v=v_1,v_2,\ldots,v_n=u)$ such
      that $vu$ is incident to the outer face and for every $j=2,\ldots,n$,
      the two edges incident to $v_j$ in the Hamiltonian cycle, 
      i.e., edges $v_{j-1}v_j$ and $v_{j+1}v_j$, are consecutive in the rotation of $v_j$
     in the subgraph induced by $v_1,\ldots,v_j,v_{j+1}$ in $G$. 
     In particular, the one-sided condition depends on the Hamiltonian cycle, its direction and
     its special edge.
A more visual reformulation of the second condition is obtained using the closed
      bounded region $D$ whose boundary is the Hamiltonian cycle. It is that
      in the embedding of $G$ for every $j$ either all the back-edges $v_iv_j$
      with $i < j$ are drawn inside $D$ or in the open exterior of $D$. We let
      $V_I$ be the set of vertices~$v_j$ which have a back-edge $v_iv_j$ with
      $i < j-1$ drawn inside $D$ and $V_O = V \setminus V_I$.  The set $V_I$
      is the set of vertices having back-edges only inside $D$ while vertices
      in $V_O$ have back-edges only outside~$D$.

   \calc_figscale{32}
    \begin{figure}[htb]
    \centerline{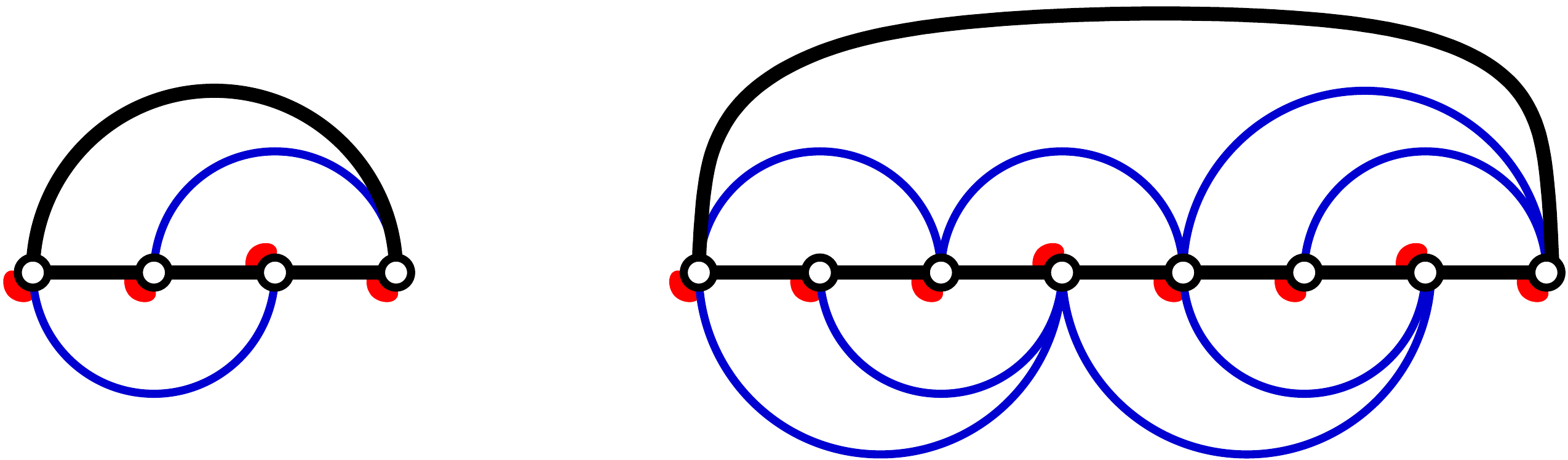}
    \caption{\label{fig:one-sided2}}
    \end{figure}
    VC
{ $K_4$ and a slightly larger graph both with a one-sided Hamiltonian cycle. 
Red angles indicate a side with no back-edge. }

Recall that $\cC'$ is the class of planar graphs which are
\emph{spanning} subgraphs of plane graphs admitting a one-sided Hamiltonian cycle. 
It is worth noting all subgraphs are POSH.

\begin{proposition}
    Any subgraph of a POSH graph is POSH.
\end{proposition}
\begin{proof}
    As edge deletions preserve the POSH property by definition, it suffices to show that deleting a vertex preserves it as well. Let $G$ be a POSH graph and let $G'$ be its supergraph with a one-sided Hamiltonian cycle. Now after deleting $v$ from $G'$, adding an edge between its neighbours on the Hamiltonian cycle (if it does not exist) can be done along the two edges of $v$ along the cycle. This is a supergraph of $G\setminus v$ with a one-sided Hamiltonian cycle.
\end{proof}

\section{The embedding strategy}

Our interest in POSH graphs is motivated by the following theorem.

\begin{theorem}\label{thm:embedding}
	Let $G'$ be POSH and let $v_1,\ldots,v_n$ be
	a one-sided Hamiltonian cycle of a plane supergraph $G$ of $G'$ on the same vertex set.
	Then there is a crossing-free embedding of~$G'$ on $H_n$ with the property that
	$v_i$ is placed on either $p_i$ or $q_i$.
\end{theorem}
\begin{proof}
	
	It is sufficient to describe the embedding of the supergraph $G$ on $H_n$.
	For the proof we assume that in the plane drawing of $G$ the sequence
	$v_1,\ldots,v_n$ traverses the boundary of~$D$ in counter-clockwise
	direction. For each $i$ vertex $v_i$ is embedded at $\bv_i=p_i$ if
	$v_i \in V_I$ and at $\bv_i=q_i$ if $v_i \in V_O$.
	
	Let $G_i = G[v_1,\ldots,v_i]$ be the subgraph of $G$ induced by
	$\{v_1,\ldots,v_i\}$.  The path $\Lambda_i=v_1,\ldots,v_i$ separates $G_i$.
	The \emph{left part} $GL_i$ consists of the intersection of $G_i$ with $D$,
	the \emph{right part} $GR_i$ is $G_i$ minus all edges which are interior to~$D$. 
	The intersection of $GL_i$ and $GR_i$ is~$\Lambda_i$ and their union is $G_i$. 
	The counter-clockwise boundary walk of $G_i$ consists of a path
	$\partial R_i$ from $v_1$ to $v_i$ which is contained in $GR_i$ and a
	path from $v_i$ to $v_1$ which is contained in $GL_i$, let $\partial L_i$
	be the reverse of this path. 
	
	Let $\bG_i$ be the straight-line drawing of the plane graph $G_i$ obtained by
	placing each vertex~$v_j$ at the corresponding $\bv_j$. A vertex $\bv$ of
	$\bG_i$ is said to \emph{see a point} $p$ if there is no crossing between
	the segment $\bv p$ and an edge of $\bG_i$.  By induction on $i$ we show:
	\begin{enumerate}
		\item\label{itm1}
		The drawing $\bG_i$ is plane, i.e., non-crossing.\smallskip
		\item\label{itm2}
		$\bG_i$ and $G_i$ have the same outer boundary walks.\smallskip
		\item\label{itm3}
		Every vertex of $\partial L_i$ in $\bG_i$ sees all the points $p_j$ with $j>i$
		and every vertex of $\partial R_i$ in $\bG_i$ sees all the points $q_j$ with $j>i$.
	\end{enumerate}
	
	For $i=2$ the graph $G_i$ is just an edge and the three claims are immediate,
	for Property~\ref{itm3} just recall that the line spanned by $p_1$ and $p_2$ separates
	the $p$-side and the $q$-side of $H_n$.
	
	Now assume that $i \in \{3,\ldots,n\}$, the properties are true for $\bG_{i-1}$ and
	suppose that $v_{i}\in V_I$ (the argument in the case $v_i \in V_O$ works symmetrically).
	This implies that all the back-edges of $v_i$ are in the interior of $D$ whence all
	the neighbors of $v_i$ belong to $\partial L_{i-1}$. Since $v_{i}\in V_I$ we
	have $\bv_i = p_i$ and Property \ref{itm3} of $\bG_{i-1}$ implies that the edges
	connecting to $\bv_i$ can be added to $\bG_{i-1}$ without introducing a crossing.
	This is Property \ref{itm1} of $\bG_{i}$.
	
	Since $G_{i-1}$ and $\bG_{i-1}$ have the same boundary walks and $v_i$ (respectively~$\bv_i$)
	belong to the outer faces of $G_{i}$ (respectively $\bG_{i}$) and since $v_i$
	has the same incident edges in $G_i$ as $\bv_i$ in $\bG_{i}$, 
	the outer walks of $G_{i}$ and $\bG_{i}$ again equal each other, i.e., 
	Property \ref{itm2}.
	
	Let $j$ be minimal such that $v_jv_i$ is an edge and note that $\partial L_i$
	is obtained by taking the prefix of $\partial L_{i-1}$ whose last vertex is
	$v_j$ and append $v_i$. The line spanned by $\bv_j$ and $\bv_i=p_i$ separates
	all the edges incident to $\bv_i$ in $\bG_i$ from all the segments
	$\bv_\ell p_k$ with $\ell < j$ and $\bv_\ell\in \partial L_i$ and $k>i$. This
	shows that every vertex of $\partial L_i$ in $\bG_i$ sees all the points $p_k$
	with $k>i$.  For the proof of the second part of Property \ref{itm3} 
	assume some edge $\bv_i\bv_j$ crosses the line of sight from $\bv_l$ to $q_k, k>i$, we refer
	to Figure~\ref{fig:unobst}. 
	First note that this is only possible if $l\leq j$, 
	since otherwise $\bv_j\bv_l$ separates $\bv_i=p_i$ and $q_k$, 
	because $p_i$ is on the left as can be seen at $x=i$ and
	$q_k$ is on the right as can be seen at $x=k$ by definition.
	Since $j=l$ is impossible by construction, we are left with the case $l<j$.
	Then one of $\bv_i$ and $\bv_l$, say $\bv$, lies to the right of the oriented line $\bv_jq_k$. 
	However that implies that $\bv_j\bv$ has $q_k$ on its left,
	which is a contradiction to the definition of $q_k$ at $x=k$.
	This completes the proof of Property \ref{itm3} and thus the inductive step.
	
	Finally,
	Property~\ref{itm1} for $\bG_n$ implies the theorem.
\end{proof} 

   \calc_figscale{20}
    \begin{figure}[htb]
    \centerline{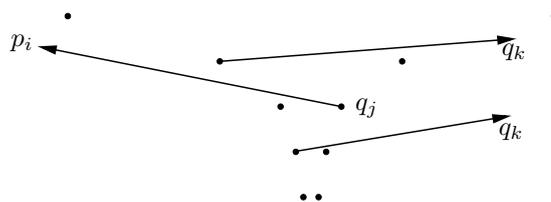}
    \caption{ Vertices from $\partial R_i$ see $q_k$\label{fig:unobst}}
    \end{figure}
    

\section{Plane bipartite graphs}
\label{sec:bipartite}

In this section we consider bipartite plane graphs and show that they are
POSH.

\begin{theorem}\label{thm:bipartite}
	Every bipartite plane graph $G=(V,E)$ is a subgraph of a plane
	graph $G'$ on the same vertex set $V$ which has a one-sided Hamiltonian
	cycle, i.e., $G$ is POSH.
\end{theorem}

\begin{proof} Quadrangulations are the plane graphs with all faces of degree
four. Equivalently they are the maximal plane bipartite graphs, i.e., any bipartite plane
graph except stars is a subgraph of a quadrangulation.
Thus since POSH graphs are closed under 
taking  subgraphs, it suffices to prove the theorem for quadrangulations.

Let $Q$ be a quadrangulation and let $V_B$ and $V_W$ be the \emph{black} and \emph{white} vertices
of a 2-coloring. Label the two black vertices of the outer face as $s$ and
$t$. Henceforth, when talking about a quadrangulation we think of an embedded
quadrangulation endowed with $s$ and $t$.
A \emph{separating decomposition} is a pair $D=(Q,Y)$ where $Q$ is a
quadrangulation and $Y$ is an orientation and coloring of the edges of
$Q$ with colors red and blue such that:
\begin{enumerate}
	\item The edges incident to $s$ and $t$ are incoming in color red and blue, respectively.
	\item Every vertex $v\not\in \{s,t\}$ is incident to a non-empty interval of red edges and
	a non-empty interval of blue edges. If~$v$ is white, then, in clockwise order,
	the first edge in the interval of a color is outgoing and all the
	other edges of the interval are incoming. If~$v$ is black, the outgoing
	edge is the clockwise last in its color (see Figure~\ref{fig:vertex-2-cond}).
\end{enumerate}
   \calc_figscale{28}
    \begin{figure}[htb]
    \centerline{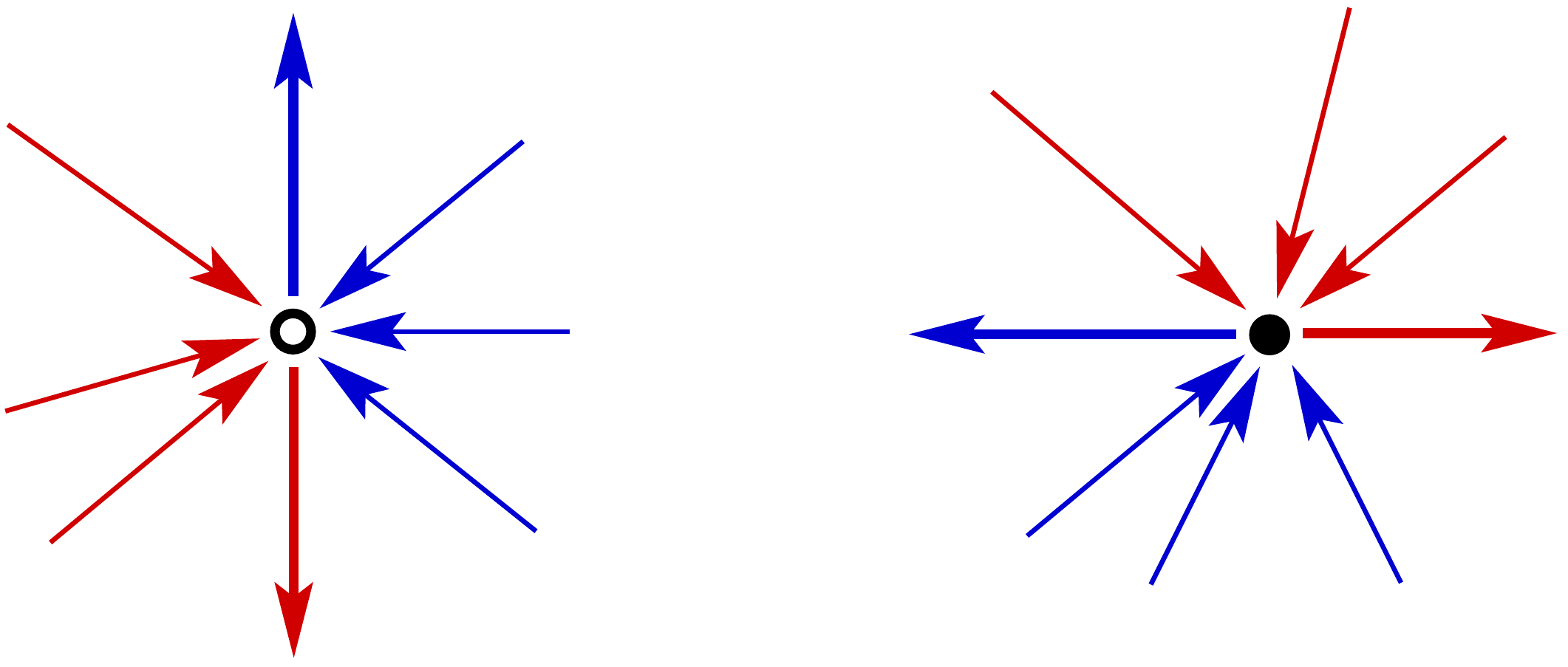}
    \caption{\label{fig:vertex-2-cond}}
    \end{figure}
    VC
{ Edge orientations and colors at white and black vertices.\cite{FelsnerFusy}}

Separating decompositions of a quadrangulation $Q$ have been defined by
de~Fraysseix and Ossona de~Mendez~\cite{deMendez}.  They show a bijection
between separating decompositions and $2$-orientations (orientations of the
edges of $Q$ such that every vertex $v\not\in \{s,t\}$ has out-degree $2$) and show the
existence of a $2$-orientation of $Q$ with an argument related to flows and
matchings. An inductive proof for the existence of separating decompositions
was given by Felsner~et~al.~\cite{FHKO10}, this proof is based on identifying
pairs of opposite vertices on faces.

In a separating decomposition the red edges form a tree
directed towards $s$, and the blue edges form a tree directed
towards~$t$. Each of the trees connects all the vertices $v\not\in \{s,t\}$ to the
respective root. 
Felsner et al.~(\cite{FelsnerFusy,FHKO10}) show that the edges of the two trees 
can be separated by a curve which starts in $s$, ends in $t$,
and traverses every vertex and every inner face of $Q$.
This curve is called the \emph{equatorial line}.

If $Q$ is redrawn such that the equatorial line is mapped to the $x$-axis
with~$s$ being the left end and $t$ being the right end of the line, then the
red tree and the blue tree become \emph{alternating trees} (\cite{FHKO10}, defined below)
drawn in the upper respectively lower half-plane defined by the $x$-axis. Note
that such a drawing of $Q$ is a 2-page book embedding, we call it an
\emph{alternating 2-page book embedding} to emphasize that the graphs drawn on the
two pages of the book are alternating trees.

   \calc_figscale{20}
    \begin{figure}[htb]
    \centerline{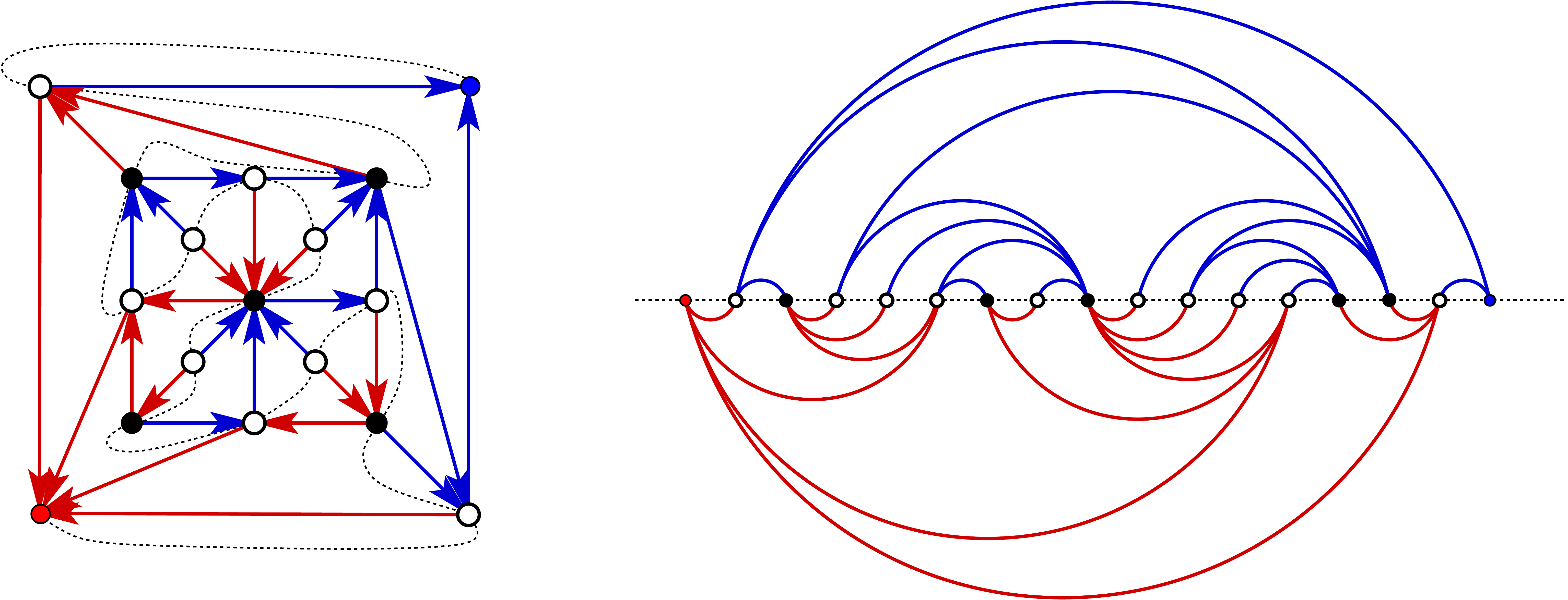}
    \caption{\label{fig:sep2bookWithEqLine}}
    \end{figure}
    VC
{  A quadrangulation $Q$ with a separating decomposition $S$, and the
alternating 2-page book embedding induced by the equatorial line of $S$\cite{FelsnerFusy}.}

An \emph{alternating tree} is a plane tree $T$
with a plane drawing such that the vertices of $T$
are placed at different points of the $x$-axis and all edges are
embedded in the half-plane above the $x$-axis (or all below).
Moreover, for every vertex~$v$ it holds that all its neighbors are
on one side, either they are all left of~$v$ or all right of $v$.
In these cases we call the vertex~$v$ respectively a \emph{right}
or a~\emph{left vertex} of the alternating layout. 
Note that every vertex is a left vertex in one of the two trees and a right vertex in the other.

Let $Q$ be a plane quadrangulation on $n$ vertices and let $S$ be a separating
decomposition of $Q$. Let $s=v_1,v_2,\ldots,v_n=t$ be the spine of the
alternating 2-page book embedding of $Q$ based on $S$.  Let $Q^+$ be obtained from
$Q$ by adding $v_nv_1$ and all the edges $v_iv_{i+1}$ which do not yet belong
to the edge set of $Q$. By construction $v_1,v_2,\ldots,v_n$ is a Hamiltonian
cycle of $Q^+$ and since the trees are alternating, black vertices have only blue edges to the left and
white vertices have only red edges to the left. Thus this
Hamiltonian cycle is one-sided with reverse edge $v_nv_1= ts$. Hence $Q$ is
POSH.
\end{proof} 

It is worth noting that the Hamiltonian cycle read in the reverse direction,
i.e., as $v_n,v_{n-1},\ldots,v_1$, is again one-sided, now the reverse edge is
$v_1v_n=st$.

\section{Planar subcubic graphs}
\label{sec:cubic}

In this section we identify another large subclass of the $\cC'$.
Recall that 3-regular graphs are also known as cubic graphs and in subcubic graphs
all vertices have degree at most 3.

\begin{theorem}\label{thm:cubic}
  Every planar subcubic graph $G$ is a spanning subgraph of a planar
  graph~$G'$ which has an embedding with a one-sided Hamiltonian cycle, i.e.,
  $G$ has a POSH embedding.
\end{theorem}
 	\begin{remark}
 		Note that we do \emph{not} claim the theorem for all {\em plane} subcubic graphs. However, we are not aware of any connected subcubic plane graph, which is not POSH.
    \end{remark}
 	
 	To prove this, we use Theorem \ref{thm:bipartite} and the following lemmas:
 	
\begin{lemma}\label{lem:cube-cont}
	Let $G$ be a subcubic graph. Then $G$ admits a matching $M$
	such that contracting all the edges of $M$ results
	in a bipartite multi-graph.
\end{lemma}

\begin{proof}
  Let $(X,Y)$ be a partition the vertex-set of $G$ such that the size of the
  cut, i.e., the number of edges in $G$ with one endpoint in $X$ and one
  endpoint in $Y$, is maximized. We claim that the induced subgraphs $G[X]$
  and $G[Y]$ of $G$ are matchings.  Suppose that a vertex $v \in X$ has at
  least two neighbors in $G[X]$. Then $v$ has at most one neighbor in $Y$, and
  hence moving $v$ from $X$ to $Y$ increases the size of the cut by at least
  one, a contradiction. The same argument works for $G[Y]$.
	
	Let $M$ be the matching in $G$ consisting of all the edges in $G[X]$ and
	$G[Y]$. Contracting the edges in $M$ transforms $G[X]$ and $G[Y]$ into
	independent sets, and hence results in a bipartite multi-graph $G/M$.
\end{proof}

A \emph{separating $k$-cycle} of a plane graph $D$ is a simple cycle of length
$k$, i.e., $k$ edges, such that there are vertices of $D$ inside the cycle.

\begin{lemma}\label{lem:cube-good}
  Let $G$ be a subcubic planar graph. Then $G$ admits a plane
  embedding $D_G$ and a matching $M$ such that contracting
  all the edges of $M$ in $D_G$ results in a bipartite multi-graph
  without separating 2-cycles.
\end{lemma}

\begin{proof} Let $G$ be a subcubic planar graph. Without loss of generality
$G$ is connected, otherwise we just deal with the components first, 
then embed $G$ in a way that all components are incident to the outer face.

Note that a 2-cycle can only arise by contracting one matching edge of a
triangle or two matching edges of a quadrilateral.
  Consider an embedding $D$ of $G$ which minimizes
  the number of separating 3-cycles and among those
  minimizes the number of separating 4-cycles. 

 \begin{Claim}
     $D$ has no separating 3-cycle.
 \end{Claim}

  \begin{proof}
      For illustration, see Figure \ref{fig:sepTri}. 
  We will first show $D$ has no \emph{separating diamond}, that is, two triangles sharing an edge $e=uv$, at least one of which is a separating 3-cycle.
  Otherwise place $u$ very closely to $v$.
  Now $e$ is short and we reroute the other two edges of $u$ such that they stay close to 
  the corresponding edge of $v$. 
  Since one of the triangles containing $e$ was assumed to be separating the 
  new drawing has fewer separating 3-cycles,
  a contradiction.

\begin{figure}[htb]
    \centering
    \includegraphics{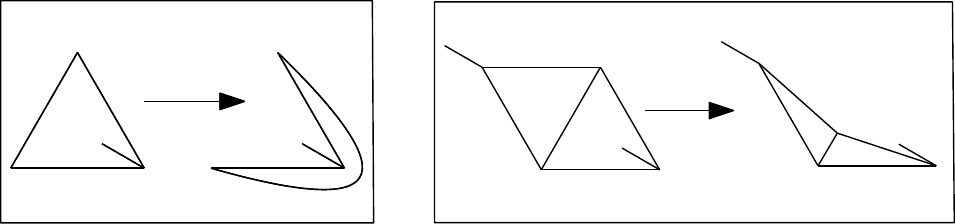}
    \caption{ Procedure to eliminate triangles with an inner vertex. The procedure on the left eliminates isolated separating triangles, while the one on the right deals with separating diamonds.}
    \label{fig:sepTri}
\end{figure}

  We are ready to show $D$ has no separating 3-cycle. 
  If $T$ is a separating 3-cycle some edge has to go from a vertex $v$ of $T$ into its interior. 
  Since $v$ has degree at most 3 it has no edge to the outside of $T$.
  We can then redraw the edge $e$ of $T$ not incident to $v$ 
  outside of $T$ closely to its two other edges. 
  Again the new drawing has fewer separating 3-cycles: indeed, if the redrawn edge would 
  be part of another 3-cycle, $T$ is part of a separating diamond.
  \end{proof}

%

Now choose an edge set $M$ of minimum cardinality, such that contracting it yields a bipartite multi-graph. The proof of Lemma \ref{lem:cube-cont} implies that $M$ is a matching. Among those matchings, we choose $M$ such that the number of separating 4-cycles
which have 2 edges in $M$ is minimized. Such separating 4-cycles are said to be \emph{covered} by $M$.

\begin{Claim}
     $M$ covers no separating 4-cycle.
\end{Claim}
\begin{proof}
Suppose $Q=v_1v_2v_3v_4$ is a separating 4-cycle
such that $v_1v_2$ and $v_3v_4 \in M$ and $v_1$ has an edge $e_I$ to the inside, thus no edge to the outside.

\begin{figure}[htb]
    \centering
    \includegraphics{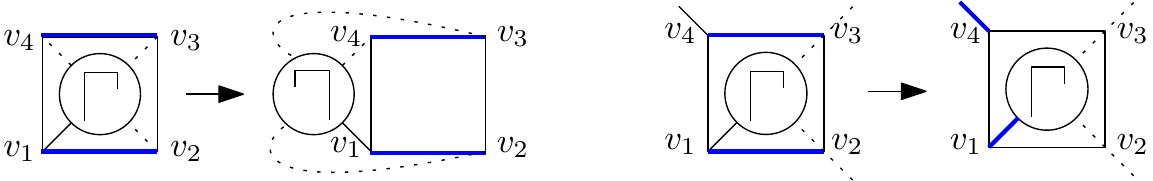}
    \caption{ Procedure to eliminate quadrilaterals with an inner vertex. The redrawing (left) cannot be applied in the right case, where we are changing the blue matching to avoid a separating 2-cycle.}
    \label{fig:sepQua}
\end{figure}

If $v_4$ has no edge to the outside either, we change $D$ to a drawing $D'$ 
by redrawing the part $\Gamma$ of $D$ inside $Q$ outside of it reflected 
across $v_1v_4$, see Figure \ref{fig:sepQua}. 
In $D'$ the original separating 4-cycle is no longer separating. 
We claim that no new separating 3-cycle or 4-cycle that is covered by $M$ was created.  
The claim contradicts the choice of $D$ or $M$.

To prove the claim note that $S=\{v_2,v_3\}$ is a 2-separator, unless $Q$ is the outer face of $D$,
so let's assume first that it is not.
Thus a separating 3- or 4-cycle has to live on one side of $S$, 
since the shortest path between them in $Q\cup \Gamma$ except their edge is of length 3
except if both $v_2$ and $v_3$ are adjacent to the same vertex of $\Gamma$,
in which case $Q$ is the outer face, a contradiction.
Let $X$ be the component of $G\setminus S$ containing $\Gamma$.
Then the number of vertices inside 3- or 4-cycles that are not part of $X$ is unchanged in $D'$,
since the face $X$ is located in is still the same.
The only 3- or 4-cycles in $X\cup S$ that were not reflected in their entirety are the ones containing the edge $v_2v_3$.
Since $Q$ is assumed not to be the outer face, at least one of $v_2$ and $v_3$ is not connected to $\Gamma$.
Thus such a cycle $C$ is a 4-cycle consisting of $v_2,v_3$, one of $v_1$ or $v_4$ 
as well as a common neighbour of $v_2$ and $v_4$ or $v_1$ and $v_3$ in $\Gamma$. 
However $v_1v_2$ or $v_3v_4$ respectively would be the only edge in $M\cap C$.
This is a contradiction to the fact that contracting $M$ yields a bipartite graph.

Now if $Q$ is the outer face of $D$, it is still true that the only cycles not reflected in their entirety contain $v_2v_3$.
However $v_2$ and $v_3$ could both be adjacent to a vertex in $\Gamma$, 
either a common neighbour for a 3-cycle or two adjacent neighbours for a 4-cycle.
Since $v_2$ and $v_3$ are already covered by $M$, this 3-cycle would contain no edge in $M$, 
whereas the 4-cycle would contain at most one.
Therefore both of these contradict the definition of $M$.

Therefore, we know that $v_4$ has an edge $e_O$ to the outside. 
This edge does not go to any vertex of the quadrilateral, because the only candidate left would be $v_2$,
but this would yield that one of the triangles $v_2v_3v_4$ and $v_1v_2v_4$ is separating.


Change the matching $M$ to an edge set $M'$ by
removing $v_1v_2$ and $v_3v_4$ from it and adding $e_O$
and $e_I$. Contracting $M'$ still results in a bipartite
graph, because the same four facial cycles that contained our previous edges
contain exactly one new edge each as well, so their size after
contraction does not change. Thus $M'$ is a matching, because it has the same cardinality as $M$ and is 
therefore minimal as well. 
We conclude $M'$ does not cover $v_2$ or $v_3$, 
because $M$ did not contain any other edge than $v_1v_2$ and $v_3v_4$ at them either.
Since $M'$ does not contain two edges from
quadrilateral $v_1,\ldots,v_4$ but $M$ is minimal, there has to be a
separating quadrilateral, of which $M'$ contains two edges, but $M$ doesn't. 
If such a separating quadrilateral $Q$ contains $e_I$, 
then it has to contain another edge incident to $v_1$. It cannot contain $v_1v_2$, 
because we know $v_2$ is not covered by $M'$. Therefore it contains $v_1v_4$ and consequently $e_O$.
The same argumentation works to show that if it contains $e_O$, then it also contains $e_I$. 
This is a contradiction to the existence of $M'$ 
because the endpoints of $e_O$ and $e_I$ are on the outside and the inside of the quadrilateral respectively
and therefore non-adjacent.
\end{proof}

So we proved that our choice of $M$ makes sure that no
separating 2-cycles will be present in the contracted plane bipartite multi-graph.
\end{proof}

  \begin{remark}
  The embedding $D$ and the matching $M$ can be constructed starting from an arbitrary
  embedding and matching by iterative application of the operations used in the proof.
  \end{remark} 
  
\begin{proof}[Proof of Theorem \ref{thm:cubic}]
Now let~$B$ be the plane bipartite multi-graph
obtained from $G$ by contracting the edges in~$M$ without changing the embedding any further.
Let $B'$ be the underlying
simple graph of $B$ and let $Q$ be a quadrangulation or a star which has $B'$ as a
spanning subgraph. The proof of Theorem~\ref{thm:bipartite} shows that there
is a left to right placement $v_1,\ldots,v_s$ of the vertices of $Q$ on the
$x$-axis such that for each $i\in[s]$ all the edges $v_jv_i$ with $j < i-1$
are in one half-plane and all edges $v_iv_j$ with $j > i+1$ are in the other
half-plane. Delete all the edges from~$Q$ which do not belong to $B'$, and
duplicate the multi-edges of $B$ in the drawing. This yields a 2-page book embedding
$\Gamma$ of $B$.

\begin{figure}[htb]
    \centering
    \includegraphics{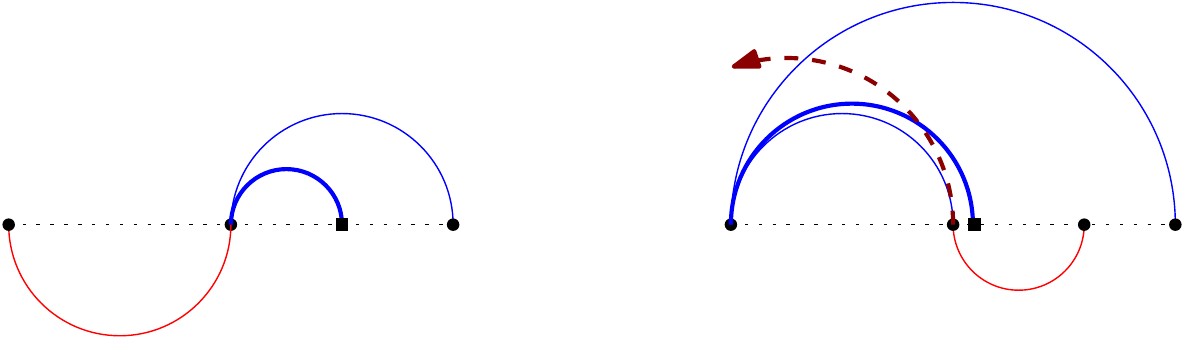}
    \caption{ How to add leaves: 
                The leaf is plotted as a square, its new adjacent edge fat.}
    \label{fig:Leaves}
\end{figure}

Let $v$ be a contracted vertex of $B$. Vertex $v$ was obtained by
contracting an edge $uw \in M$. If $u$ and/or $w$ did not have degree 3, we add 
edges at the appropriate places into the embedding that end in leaves, see Figure~\ref{fig:Leaves}. 
To add an edge to $u$ for instance, choose a face $f$ incident to $u$ that is not contracted into a 2-cycle.
Let $e$ and $e'$ be the two edges incident to both $v$ and $f$.
If the angle between $e$ and $e'$ contains part of the spine (the \hbox{$x$-axis}), 
we put the leaf on the spine close to $v$ connected to $v$ with a short edge below or above the spine, 
in a way to accomodate the local vertex condition of $v$.
If it doesn't, assume without loss of generality it is in the upper half-plane and
that edge $e$ is the edge closer to the spine.
This edge is unique because both edges at $v$ delimiting $f$ go upwards and 
therefore both to the same side, say right of $v$.
Route the new edge closely along $e$ then put the leaf just next to the other endpoint $x$ of $e$. 
Edges that would cross this new edge cannot cross $e$, 
thus the only possibility are edges incident to $x$ that emanate into the upper halfspace. 
However those edges have to go to the left of $x$ by its local vertex condition.
These edges do not exist, as any such edge would have to cross $e'$, 
see the dashed line in Figure Figure~\ref{fig:Leaves}. Thus the new edge is uncrossed.
This procedure will be done to every vertex first. 
Note that the resulting graph stays bipartite and the local vertex conditions are still fulfilled, 
but now every contracted vertex has degree 4. This makes the case distinction of splitting the vertices easier.

We now show how to undo the contractions, i.e.,
\emph{split} vertices, in the drawing $\Gamma$ in such a way
that at the end we arrive at a one-sided 2-page book drawing~$\Gamma^\star$ of $G$, 
that is, a $2$-book embedding of $G$ with vertex-sequence $v_1,\ldots,v_n$ such
that for every $j \in \{1,\ldots,n\}$ the incident back-edges $v_iv_j$ with $1 \le i<j$
are all drawn either on the spine or on the same page of the book embedding (all above
or all below the spine). Once we have obtained such a book embedding, we can delete the artificial added leaves,
then add the spine edges (including the back edge from the rightmost to the leftmost vertex)
to $G$ to obtain a supergraph $G^+$ of $G$ which has a
one-sided Hamiltonian cycle, showing that $G$ is POSH.

Before we advance to show how we split a single vertex $v$ of degree four into an edge $uw\in M$,
we first want to give an overview of the order in which the different splits,
the \emph{far splits} and \emph{local splits} are applied.
We will then describe what these different splits actually mean.
To split all the degree four vertices we proceed as follows:

First we split all vertices which are subject to a far split, 
from the outside inwards. More precisely, define a partially ordered set on the edges
incident\footnote{ There will be a clarification later as to what this means exactly.} 
to vertices subject to a far split in the following way:
Every edge $e$ defines a region $R_e$ which is enclosed by $e$ and the spine.
Now order the edges by the containment order of regions $R_e$.
From this poset, choose a maximum edge and then a vertex that needs a far split incident to that edge.
When no further far split is possible we do all the local splits.
These splits are purely local, so they cannot conflict with each other.
Therefore their order can be chosen arbitrarily.

We label the edges of $v$ in clockwise order as
$e_1,e_2,e_3,e_4$ such that in~$G$ the edges $e_1,e_2$ are incident to $u$ and
$e_3,e_4$ are incident to $w$. If the two angles $\angle e_2e_3$ and
$\angle e_4e_1$ together take part of both half-planes defined by the
spine, then it is possible to select two points left and right of the
point representing $v$ in $\Gamma$ and to slightly detour the edges $e_i$ such
that no crossings are introduced and one of the two points is incident to
$e_1,e_2$ and the other to $e_3,e_4$. The addition of an edge connecting the two
points completes the split of $v$ into the edge $uw \in
M$. Figure~\ref{fig:uncontract1} shows a few instances of this \emph{local}
split.

   \calc_figscale{22}
    \begin{figure}[htb]
    \centerline{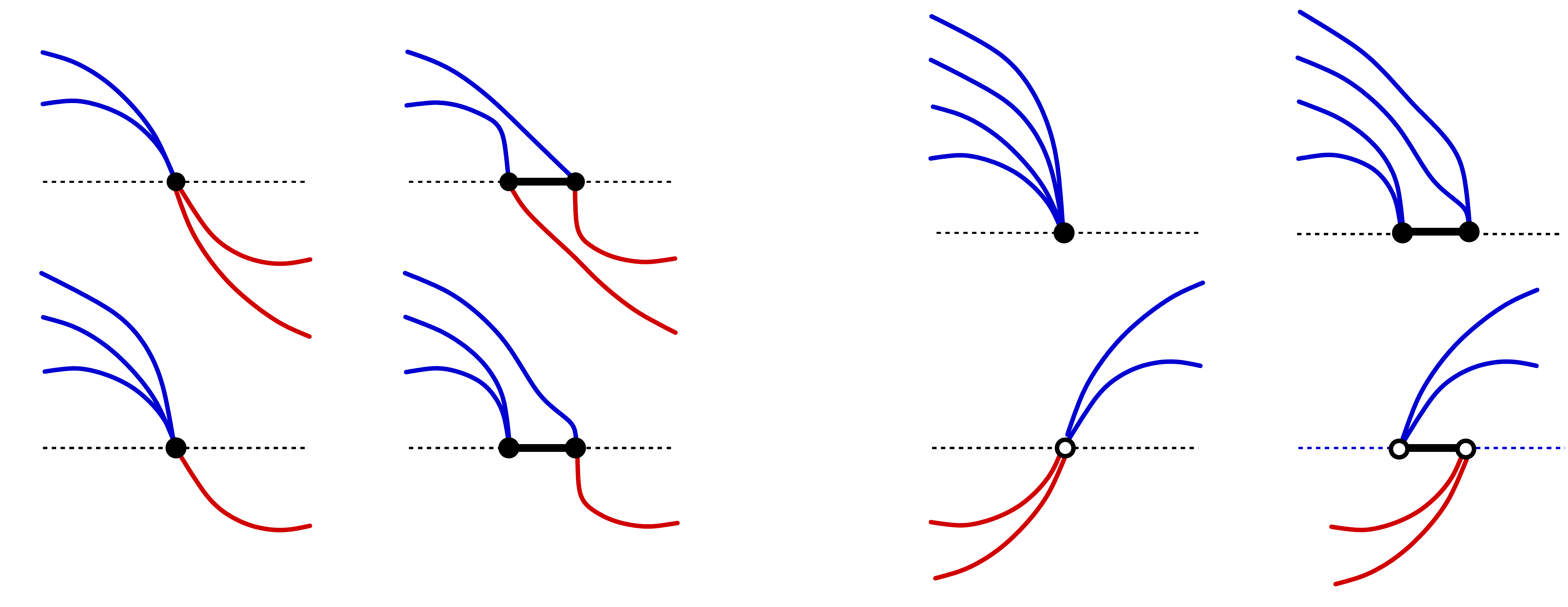}
    \caption{\label{fig:uncontract1}}
    \end{figure}
    VC
{ Four cases for the local split of a vertex $v$.}

The above condition about the two angles is not fulfilled if and only if 
all four edges of
$v$ emanate into the same halfspace, say the upper one,
and the clockwise numbering starting at the $x$-axis is either $e_4,e_1,e_2,e_3$ or $e_2,e_3,e_4,e_1$. 
The two cases are the same up to exchanging the names of $u$ and $w$, 
therefore we can assume the first one.
A more important distinction is whether most $e_i$ end to the left or right of $v$.
Note that in the ordering given by $\Gamma$, 
all $e_i$ go to the same side, since they are all in the same halfplane.
However, if $v$ is not the first vertex we are splitting, it may happen, 
that a single edge on the spine to the other side exists, see Figure~\ref{fig:uncontract3+dotted}.
For all $i\in[4]$ let $v_i$ be the other endpoint of $e_i$ than $v$. While it can happen
that some of the $v_i$ coincide due to multi-edges, we will first discuss the case that they don't.
In the left case we put $u$ slightly left of $v_1$ while in the right case $u$ is
put slightly right of $v_2$, connecting $u$ to this close vertex by a spine edge. 
In both cases we leave $w$ at the former
position of $v$.
Figure~\ref{fig:uncontract3+dotted} shows the right case and Figure~\ref{fig:uncontract2+dotted} the left.

\begin{figure}[htb]
    \centering
    \includegraphics{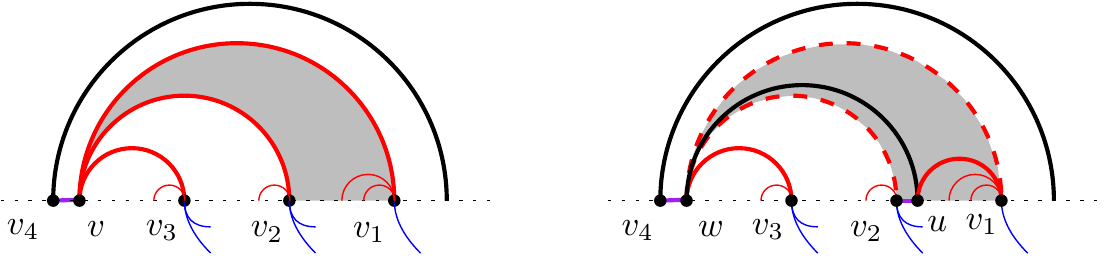}
    \caption{ Far split with $v_i$ to the right except for the spine edge neighbor.}
    \label{fig:uncontract3+dotted}
\end{figure}

To see that in the left case edges $uv_2$ and $uw$
are completely free of crossings, observe that we can route them close to the path $v_2vv_1$ and the
edge $v_1v$ respectively in the original drawing (dashed in Figure~\ref{fig:uncontract2+dotted}). 
It is important to note here, that due to the order in which we chose to do the splits, 
$v_1$ and $v_2$ are still original vertices of $B$, that is, they have not been split in the upper half-plane and thus 
still don't have two edges emanating to the upper half-plane to both sides.
Therefore, similarly to the argumentation for adding leaves, 
no edge incident to $v_1$ crosses $uw$ or $uv_2$.
The right case is analogous, just exchange the roles of $v_1$ and $v_2$.

This kind of split is a
\emph{far} split. For the purposes of incidence in the poset structure mentioned above,
vertices are not only considered incident to any edge they are an endpoint of, 
but the spine neighbour of $u$ ($v_1$ or $v_2$)
is also considered to be incident to the edge $uw$.
For illustration, consider the outermost black edge in Figure~\ref{fig:uncontract3+dotted} (left), it is considered incident to $v$.


%
   \calc_figscale{32}
    \begin{figure}[htb]
    \centerline{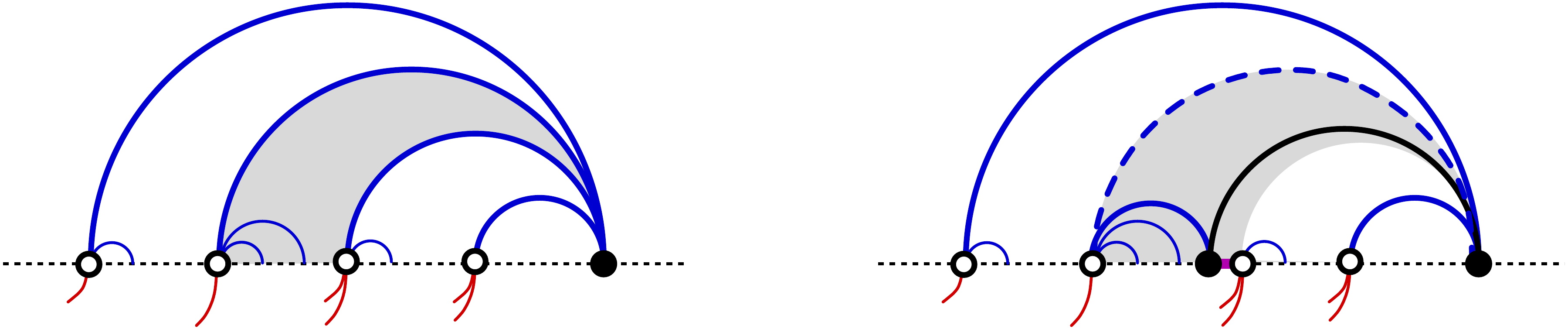}
    \caption{\label{fig:uncontract2+dotted}}
    \end{figure}
    VC
{ Far split within the gray region with $v_i$ to the left in the upper half-plane.}
%

In the following we describe how the different kinds of splits are affected by the presence of multi-edges.
The first thing to note is that local splits can be done in the same way, since we did not mention the end vertices at all.

\begin{figure}[htb]
    \centering
    \includegraphics[scale=.95]{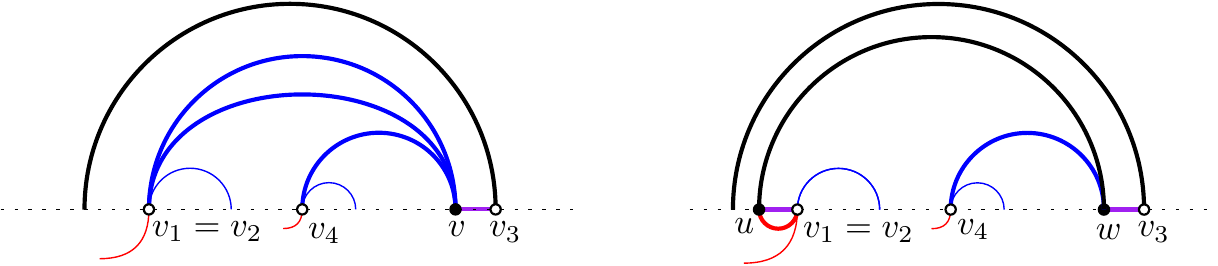}
    \caption{ If $v_1=v_2$, a double spine edge is created. Here $e_3=vv_3$ is a spine edge.}
    \label{fig:v1=v2}
\end{figure}

Concerning the far splits, firstly we talk about the case that exactly two edges go from one vertex to another:
As depicted in Figures ~\ref{fig:uncontract3+dotted} and~\ref{fig:uncontract2+dotted} the case
$v_2=v_3$ and/or $v_4=v_1$ is unproblematic, in this case we keep the dashed line(s) in the drawing. 
Double-edges are consecutive because non-consecutive double-edges are separating 2-cycles, which we avoided in the construction. 
Thus the last case of a double-edge to consider is $v_1=v_2$.
In this case, we follow the same strategy of placement of $u$ and $w$, but this results in a double-edge on the spine between $u$ and $v_1=v_2$, see Figure~\ref{fig:v1=v2}. 
As in later local splits, we might be interested what half-space the angle between the two spine edges is part of, 
we interpret one of these edges as a spine edge and the other as an edge which is above or below the spine 
depending on the right vertex of the two. 
This might be $u$ or $v_1$, depending on whether we are in the left or right case. 
It is important for the one-sidedness condition to choose this direction so that 
all left neighbours of the right vertex of the two are reached by edges emanating into the same halfspace and/or spine edges.

\begin{figure}[htb]
    \centering
    \includegraphics[scale=.9]{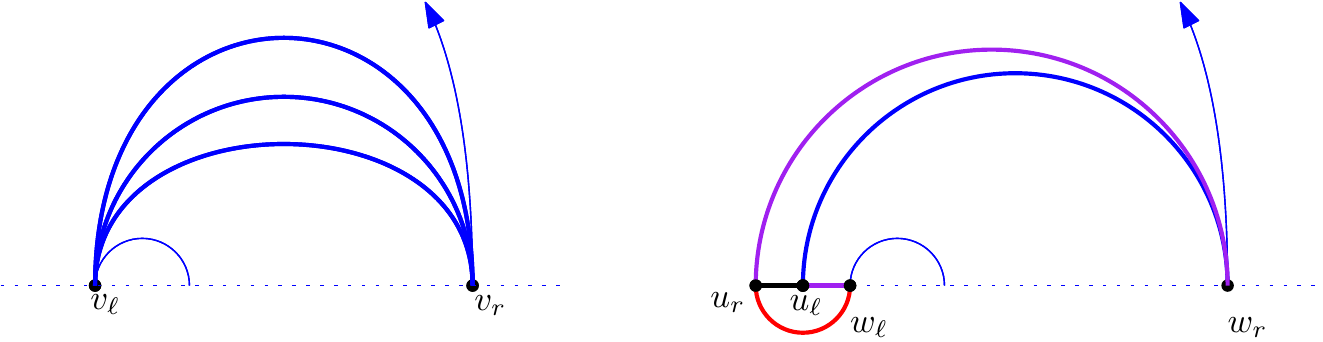}
    \caption{ Doing a double split means splitting two vertices simultaneously. }
    \label{fig:2Split}
\end{figure}

Secondly, if there are three edges between a left vertex $v_\ell$ and a right vertex $v_r$,
say in the upper half-plane, we will split both simultaneously, for illustration, see Figure~\ref{fig:2Split}. 
Since three edges go between these two vertices, there is just one more edge $e$ left for $v_\ell$.
Therefore we can find a place on the spine just to the right or to the left of $v_\ell$ which is free, because the edge $e$ is on the other side.
Now we split $v_\ell$ into $u_\ell$ and $w_\ell$ and $v_r$ into $u_r$ and $w_r$ simultaneously 
where $w_\ell$ and $w_r$ are the vertices with the edge that goes somewhere else on both sides.
From left to right we put $u_r$ then $u_\ell$ just left of the position of $v_l$, which is the new position of $w_\ell$.
The three of them are connected by spine edges, just $u_r$ and $w_\ell$ have an edge in the lower half-plane. 
These edges are not crossed, because the vertices are close enough together.
Finally we put $w_r$ at the position of $v_r$ and add edges to $w_r$ and $w_\ell$ in the upper half-plane.
These edges are not crossed, because any edge crossing them would have crossed the triple edge in the original drawing.

This kind of split is a \emph{double} split.
These splits are purely local, so they can be performed together with the local splits in the end.

The last case is that all four edges of a given vertex go to the same vertex, this is a full connected component of the bipartite graph, because it has maximum degree 4. 
This component goes back to a $K_4$ component in the cubic graph that had two independent edges contracted. A one-sided Hamiltonian cycle of $K_4$ is illustrated in Figure~\ref{fig:one-sided2}. 
We apply another local double split which consists of replacing the 4 parallel edges by this drawing, embedded close to the place of one of the original vertices.



This completes the proof of Theorem~\ref{thm:cubic}.
\end{proof}

\section{2-Trees}\label{sec:2trees}

From the positive results in Sections \ref{sec:bipartite} and \ref{sec:cubic} one might expect
that ``sufficiently sparse'' planar graphs are POSH. This
section shows that 2-trees are not.

A \emph{$2$-tree} is a graph which can be obtained, starting from a $K_3$, by
repeatedly selecting an edge of the current graph and adding a new vertex
which is made adjacent to the endpoints of that edge. We refer to this operation as
\emph{stacking} a vertex over an edge.

From the recursive construction it follows that a $2$-tree on $n$ vertices
is a planar graph with $2n-3$ edges. We also mention that $2$-trees
are series-parallel planar graphs. Another well studied class which contains
$2$-trees as a subclass is the class of (planar) Laman graphs.

Fulek and Tóth have shown that planar 3-trees 
admit $n$-universal point sets of size $O(n^{3/2}\log n)$. Since every 
$2$-tree is an induced subgraph of a planar 3-tree the bound carries over
to this class. 

\begin{theorem}
	There is a 2-tree $G$ on 499 vertices that is not POSH.
\end{theorem}

\begin{proof}
Throughout the proof we assume that a $2$-tree $G$ is given
together with a left to right placement $v_1,\ldots,v_n$ of the vertices on
the $x$-axis such that adding the spine edges and the reverse edge $v_nv_1$ to
$G$ yields a plane graph with a one-sided Hamiltonian cycle.

\def\XL{X\kern-2ptL}
\def\XM{X\kern-2ptM}
\def\XR{X\kern-2ptR}

For an edge $e$ of $G$ we let $X(e)$ be the set of vertices which are stacked
over $e$ and $S(e)$ the set of edges which have been created by stacking over
$e$, i.e., each edge in $S(e)$ has one vertex of $e$ and one vertex in $X(e)$.
We partition the set $X(e)$ of an edge $e=v_iv_j$ with $i<j$ into a left part
$\XL(e) = \{ v_k \in X(e) : k < i \}$, a middle part
$\XM(e) = \{ v_k \in X(e) : i < k < j \}$, and a right part:
$\XR(e) = \{ v_k \in X(e) : j < k \}$.

\begin{Claim}\label{claim:XR}
	For every edge $|\XR(e)| \leq 2$.
\end{Claim}
Suppose that $|\XR(e)| \geq 3$.
Each vertex in this set has all its back-edges on the same side. Two of them
use the same side for the back edges to the vertices of~$e$. This
implies a crossing pair of edges, a contradiction.

\begin{Claim}\label{claim:XM}
	If for all $e' \in S(e)$ we have $|X(e')| \geq 3$, then
	$|\XM(e)| \leq 3$.
\end{Claim}

Suppose that $e=v_iv_j$ with $i<j$ is in the upper half-plane and there are
four vertices $x_1,x_2,x_3,x_4$ in $\XM(e)$. One-sidedness implies that the
four edges $x_kv_j$ are in the upper half-plane.  Thus if $x_1,x_2,x_3,x_4$ is
the left to right order, then the edges $v_ix_2$, $v_ix_3$, and $v_ix_4$ have
to be in the lower half-plane. Now let $e'=v_ix_3$ and consider the three
vertices in $X(e')$. Two of them, say $y_1, y_2$, 
are on the same side of $x_3$.
First suppose $y_1,y_2\in X(e')$ are left of $x_3$. The edges
of $v_ix_2$ and $x_2v_j$ enforce that  $y_1,y_2$ are between $x_2$ and $x_3$. Due to
edge $x_2v_j$ the edges $v_iy_1,v_iy_2$ are in the lower half-plane.
One-sidedness at $x_3$ requires that $y_1x_3$ and $y_2x_3$ are also in the lower
half-plane. This makes a crossing unavoidable.

Now suppose that $y_1,y_2\in X(e')$ are right of $x_3$. The edges
$v_ix_4$ and $x_4v_j$ enforce that  $y_1,y_2$ are between $x_3$ and $x_4$. Due to
the edge $x_3v_j$ the edges $v_iy_1$ and $v_iy_2$ are in the lower half-plane.
Now let $y_1$ be left of $y_2$.  One-sidedness at~$y_2$ requires that $x_3y_2$
is also in the lower half-plane, whence, there is a crossing between $v_iy_1$ and $x_3y_2$.
This completes the proof of the claim.

   \calc_figscale{21}
    \begin{figure}[htb]
    \centerline{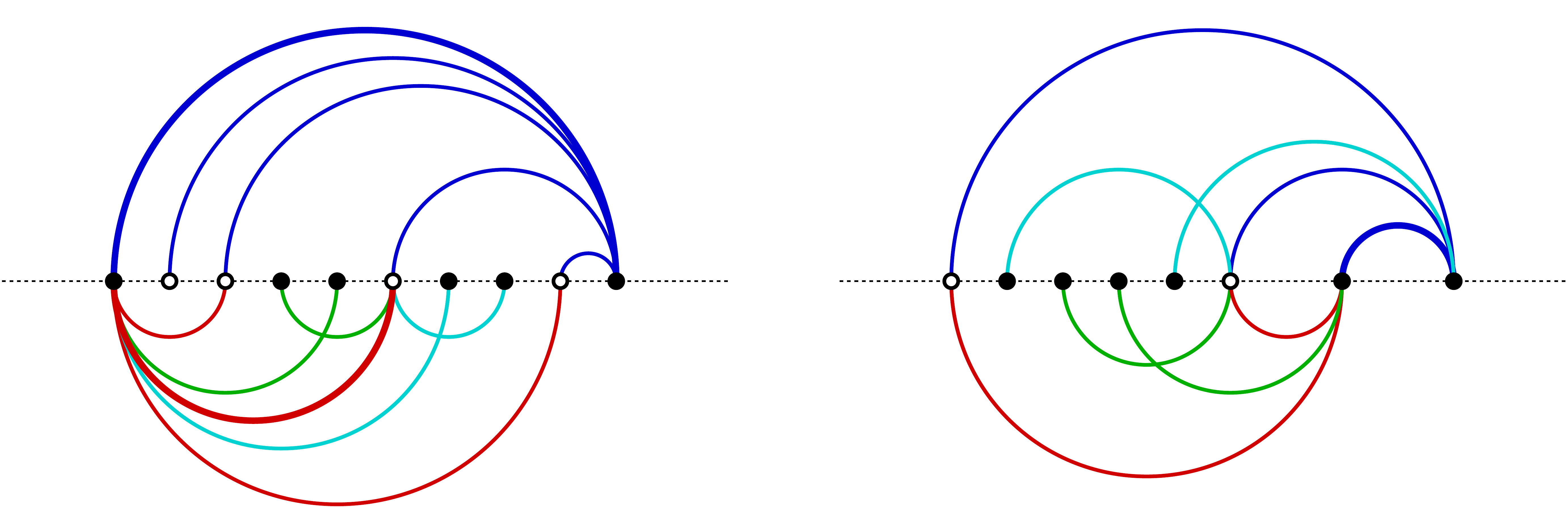}
    \caption{\label{fig:claim23}}
    \end{figure}
    VC
{ Illustrating the proofs of the claims.}

\begin{Claim}\label{claim:XL}
	If $\XL(e)\ge 2$ and $x$ is the rightmost element of $\XL(e)$, then
	$\XL(e')\le 1$ for some $e' \in S(e)$ incident with $x$ and $\XR(e')=\emptyset$ for both. 
\end{Claim}

Suppose that $e=v_iv_j$ with $i<j$ is in the upper half-plane and there are
two vertices $x_1,x_2$ in $\XL(e)$. We assume that $x_2$ is the rightmost
element of $\XL(e)$. From one-sidedness at $v_j$ we know that $x_1v_j$ and $x_2v_j$
are in the upper half-plane. Now $x_1v_i$ and hence also $x_2v_i$ are in the
lower half-plane.  All the vertices of $X(x_2v_i)$ and $X(x_2v_j)$ are in the
region bounded by $x_1v_j,v_jv_i,v_ix_1$, in particular $\XR(e')=\emptyset$ for both.
Suppose for contradiction that we have
$y_1,y_2 \in \XL(x_2v_i)$ and $z_1,z_2 \in \XL(x_2v_j)$.  By one-sidedness the
edges from $x_2$ to the four vertices $y_1,y_2,z_1,z_2$ are in the same
half-plane. If they are in the lower half-plane and $y_1$ is left of $y_2$
there is a crossing between $y_1x_2$ and $y_2v_i$.  If they are in the upper
half-plane and $z_1$ is left of $z_2$ there is a crossing between $z_1x_2$ and
$z_2v_j$.  The contradiction shows that $\XL(x_2v_i)\le 1$ or
$\XL(x_2v_j)\le 1$, since $x=x_2$ this completes the proof of the claim.  \medskip

We are ready to define the graph $G$ and then use the claims to prove that
$G$ is not POSH. The graph $G$ contains a \emph{base edge} $e$ and
seven vertices stacked on~$e$, i.e., $|X(e)|=7$. 
For each edge $e'\in S(e)$ there are five vertices stacked on $e'$.
Finally, for each edge $e''$ introduced like that three vertices are stacked on $e''$.
Note that there are $7\cdot 2 = 14$ edges $e'$, $14 \cdot 5 \cdot 2=140$ edges $e''$
and $140 \cdot 3 \cdot 2=840$ edges
introduced by stacking on an edge $e''$. In total the number of edges is $995 = 2n-3$, hence,
the graph has 499 vertices.

Now suppose that $G$ is POSH and let $v_1,\ldots,v_n$ be the order of vertices on
the spine of a certifying 2-page book embedding. Let $e=v_iv_j$ with $i<j$ be 
the base edge.   Assume  by symmetry that $e$ is in the
upper half-plane.  From Claim~\ref{claim:XR} we get $|\XR(e)| \leq 2$ and from Claim~\ref{claim:XM} we
get $|\XM(e)| \leq 3$, it follows that $|\XL(e)| \geq 2$. Let $x_1$ and $x_2$
be elements of $\XL(e)$ such that $x_2$ is the rightmost element of
$\XL(e)$. Let $e'=x_2v_i$ and $e''=x_2v_j$ then
$\XR(e') = \emptyset =\XR(e'')$ by Claim~\ref{claim:XL}. From Claim~\ref{claim:XM} applied to $e'$ and $e''$ we deduce that
$|\XM(e')| \leq 3$ and $|\XM(e'')| \leq 3$. Hence $|\XL(e')| \geq 2$ and
$|\XL(e'')| \geq 2$.  This is in contradiction with Claim 3. Thus
there is no spine ordering for $G$ which leads to a one-sided crossing-free
2-page book embedding.
\end{proof}

\section{Concluding remarks}\label{sec:conclusion}


We have examined the exploding double chain as a special point set 
(order type) and shown that the initial part $H_n$ of size $2n-2$
is $n$-universal for graphs on $n$ vertices that are POSH.
We believe that the class of POSH graphs is quite rich. 
On the sparse side, the result on bipartite graphs might be generalized, while for triangulations, 
the sheer number of Hamiltonian cycles in 5-connected graphs \cite{Alahmadi202027} makes it likely
one of them is one-sided.

\begin{conjecture}
	Every triangle-free planar graph is POSH.
\end{conjecture}

\begin{conjecture}
	Every 5-connected planar triangulation is POSH.
\end{conjecture}

We have shown that $2$-trees and their superclasses series-parallel and
planar Laman graphs are not contained in the class $\cC'$ of POSH graphs.
The question whether these classes admit universal point sets of linear size
remains intriguing. 

\bibliography{arxiv}

\end{document}